\documentclass[UKenglish]{lipics-v2019}

\title{Time Space Optimal Algorithm for Computing Separators in Bounded Genus Graphs}

\author{Chetan Gupta}{Indian Institute of Technology Kanpur,  India}{gchetan@cse.iitk.ac.in}{}{Visvesvaraya PhD Grant}

\author{Rahul Jain}{Indian Institute of Technology Kanpur, India}{jain@cse.iitk.ac.in}{}{Ministry of Human Resource Development, Government of India}

\author{Raghunath Tewari}{Indian Institute of Technology Kanpur, India}{rtewari@cse.iitk.ac.in}{}{DST Inspire Faculty Grant and Visvesvaraya Young Faculty Fellowship}

\authorrunning{C. Gupta, R. Jain and R. Tewari}

\Copyright{Chetan Gupta, Rahul Jain and Raghunath Tewari}



\nolinenumbers

\hideLIPIcs 

\usepackage{tikz}
\usepackage{complexity}
\usetikzlibrary{automata,positioning, decorations, arrows}

\newcommand{\dist}[2]{\textsf{dist}(#1, #2)}
\newcommand{\vor}[1]{\textsf{vor}(#1)}

\newcommand{\core}[1]{\textsf{core}(#1)}

\newcommand{\mframe}[1]{\textsf{mframe}(#1)}
\newcommand{\fgraph}[1]{\textsf{frame}(#1)}
\newcommand{\pfcycle}[1]{\textsf{pfloop}(#1)}
\newcommand{\boss}[1]{\textsf{boss}(#1)}
\newcommand{\head}[1]{\textsf{head}(#1)}
\newcommand{\rev}[1]{\textsf{rev}(#1)}
\newcommand{\tail}[1]{\textsf{tail}(#1)}
\newcommand{\fcycle}[1]{\textsf{fcycle}(#1)}
\newcommand{\con}[1]{\textsf{con}(#1)}
\newcommand{\pfloop}[1]{\textsf{pfloop}(#1)}
\newcommand{\smext}[1]{\textsf{smext}(#1)}
\newcommand{\smint}[1]{\textsf{smint}(#1)}
\newcommand{\ind}[1]{\textsf{ind}(#1)}
\newcommand{\first}[1]{\textsf{first}(#1)}
\newcommand{\last}[1]{\textsf{last}(#1)}

\newcommand{\lface}[1]{\textsf{left}(#1)}
\newcommand{\rface}[1]{\textsf{right}(#1)}
\newcommand{\nrst}{\textsf{nrst}}

\begin{document}
	
	\maketitle
	
	\begin{abstract}
		A graph separator is a subset of vertices of a graph whose removal divides the graph into small components. Computing small graph separators for various classes of graphs is an important computational task. In this paper, we present a polynomial time algorithm that uses $O(g^{1/2}n^{1/2}\log n)$-space to find an $O(g^{1/2}n^{1/2})$-sized separator of a graph having $n$ vertices and embedded on a surface of genus $g$.
	\end{abstract}
	\section{Introduction}
	\label{sec:intro}
	Graph separator is a useful tool in designing divide and conquer based algorithms for various graph problems. In a graph, a separator is a small set of vertices of the graph whose removal divides the graph into pieces such that the size of each piece is at most a fraction of the original graph. Lipton and Tarjan's pioneering result showed that there exists a separator of size $O(n^{1/2})$ in planar graphs. Subsequently, this separator was used to design a large number of algorithms to solve various problems in planar graphs.
	
	Recently, researchers have been interested in designing memory-constrained algorithms for various graph problems. The goal in designing such algorithms is to optimize the space required by the algorithm while maintaining the polynomial time-bound. Graph separators have been used in designing memory-constrained algorithms for the \textit{reachability} problem. Imai et al. and Ashida et al. presented polynomial-time algorithms that use $O(n^{1/2}\log n)$ space to find a separator of size $O(n^{1/2})$ in a planar graph  \cite{Imai,Imaiplanarsep}. Imai et al. also gave a memory-constrained algorithm to solve the reachability problem using this separator \cite{Imai}. A natural extension planar graphs is the set of graphs that we can embed on a surface of constant \textit{genus}. For such graphs, we know that a separator of size $O(n^{1/2})$ exists. Chakraborty et al. gave a polynomial time algorithm which uses $O(n^{2/3}\log n)$ space to construct a separator of size $O(n^{2/3})$ in constant genus graphs \cite{fsttcstewari}.
	
	Jain and Tewari formalized the connection between separators in a class of undirected graphs and the reachability problem in the class of directed versions of those graphs \cite{mytreewidth}. They mainly show that if there exists a polynomial-time algorithm that uses $O(w\log n)$-space to find a separator of size $O(w)$, then there exists a polynomial-time algorithm that uses $O(w\log n)$ space to solve reachability as well.
	
	In this paper, we continue along the above line of work and present a polynomial-time algorithm that uses $O(g^{1/2}n^{1/2}\log n)$ space to construct a separator of size $O(g^{1/2}n^{1/2})$ in a $g$-genus graph. Therefore, combining this construction with \cite{mytreewidth} we get an $O(g^{1/2}n^{1/2} \log n)$ space and polynomial time algorithm to solve the reachability problem in $g$-genus graphs. Thus, for constant genus graphs, our approach gives a polynomial time algorithm that uses $(n^{1/2}\log n)$ space.
	
	Our construction of separator follows standard paradigm used in previous constructions of separators for planar graphs and surface-embedded graphs. Hence some familiarity with earlier results such as those shown by Gazit and Miller $\cite{garrymiller}$, Koutis and Miller $\cite{koutismiller}$, Imai et al. $\cite{Imai}$, Ashida et al. $\cite{Imaiplanarsep}$, and Chakraborty et al. $\cite{fsttcstewari}$ is beneficial in understanding our construction. In particular, since our result is a generalization of Ashida et al. \cite{Imaiplanarsep}, we heavily borrow their framework.
	
	\subsection*{Our Result}	
	In this paper, we prove the following theorem.
	
	\begin{theorem}
		\label{thm:main}
		Given a $g$-genus graph $G$ of $n$ vertices, there exists a polynomial time algorithm that uses $O(g^{1/2}n^{1/2} \log n)$ space to find a separator of size $O(g^{1/2}n^{1/2})$ in $G$.
	\end{theorem}
	To achieve this, given a graph $G$, we first find a maximal set of vertices in $G$ whose \textit{$k$-neighbourhoods} do not intersect each other. We call the vertices in this set \emph{Boss vertices}. We associate each vertex of the graph to one of the Boss vertices. We call the set of vertices associated with a common Boss vertex a \textit{Voronoi region}.
	If there exists a non-contractible cycle of length $O(k)$ in the graph that spans at most two of these Voronoi regions, we find that cycle and remove it from the graph. Removal of such a non-contractible cycle from the graph reduces its genus by at least one. Otherwise, if there exists no such cycle, we proceed by dividing the original graph further into a total of at most $O(n/k + g)$ regions so that each of them is bounded by a simple cycle of length $O(\sqrt{k})$ and the number of vertices present inside each region is at most $n/3$. We then use these regions to construct a \emph{Frame Graph}, which is the graph induced by the vertices on these cycles. We assign weights to each face of the frame graph such that the weight of a face is equal to the number of vertices of the original graph \emph{inside} the corresponding cycle. In the course of the construction of frame graph, we might encounter some properties of the original graph which allows us to output either a separator or a non-contractible cycle of size $O(k)$. Thus, in the end, we either have a non-contractible cycle of size $O(k)$, a small separator or the frame graph. If the result is a separator, then we output that separator. If the output is a non-contractible cycle $C$, we store it and restart the algorithm with the graph $G\setminus C$ as the input. The final output will be the union of $C$ with the separator of $G\setminus C$. If the result is the frame graph, we use the algorithm of Gilbert et al. \cite{Gilbert} to find a separator. For an appropriate value of $k$, our algorithm achieves the desired time and space bound.
	
	\subsection*{Comparison with the previous result} 
	
	The construction of separator by Chakraborty et al. $\cite{fsttcstewari}$ proceeds (roughly) as follows: Given a graph $G$ they first find a subgraph $H$ of the graph $G$, such that removal of vertices of $H$ from $G$ makes the resulting graph $G \setminus H$ planar. Subsequently, they obtain a separator of $G \setminus H$ using the algorithm of Imai et al. $\cite{Imai}$ and add it to the vertices of $H$ to get a separator of the graph $G$. The subgraph $H$ obtained in Chakraborty et al. \cite{fsttcstewari} might not be a connected subgraph of $G$. We find a smaller separator by first finding a $2$-connected weighted subgraph of $G$ such that a weight-separator of this weighted subgraph acts as a separator of the original graph. We find this $2$-connected subgraph by using \emph{Ridge edges} which was previously used by Gazit and Miller \cite{garrymiller} and Ashida et al. \cite{Imaiplanarsep} for the case of planar graphs. We show that by first efficiently removing non-contractible cycles from Voronoi regions, ridge edges can be used in graphs of higher genus as well. We call this $2$-connected subgraph a \emph{Frame graph}. We then show that weight-separator for this Frame graph can be found efficiently and hence get our result.
	
	\subsection*{Organization of the Paper}
	\label{org}
	The rest of the paper is organized as follows. In section $\ref{sec:prelims}$, we give some preliminary notations and definitions. To construct a separator for a given bounded genus graph we first divide the graph into Voronoi regions. We explain this procedure in section $\ref{sec:voronoi}$. We further divide Voronoi regions by using pre-frame-loops in section $\ref{sec:dividingVornoi}$. In section $\ref{sec:framegraph}$, we show how to process the pre-frame-loops and construct a \emph{frame graph} and then make \textit{floor} and \textit{ceiling} modifications in this frame graph. We thus get the required subgraph. Finally, in section $\ref{sec:final}$ we put it all together to construct the separator.
	
	\section{Preliminaries}
	\label{sec:prelims}    
	
	A graph is an ordered triple $G = (V(G), E(G), \partial)$ where $V(G)$ is the set of \emph{vertices}, $E(G)$ is the set of \emph{edges} and $\partial$ is a function that assigns to each edge a pair of vertices. Let $p$ be a path. We use $\first{p}$ to denote the first edge of $p$ and $\last{p}$ to denote the last edge of $p$. In an undirected graph $G$, it is helpful to regard each edge in $E$ as a pair of directed edges, or \emph{darts}. Each dart goes from one vertex, called its \emph{tail}, to another vertex, called its \emph{head}. For a dart $e$, we use $\tail{e}$ to denote the tail of the dart, and similarly, we use $\head{e}$ to denote the head of the dart. The two darts that results from a single undirected edge are said to be \emph{reverse} of each other. If two darts $e_1$ and $e_2$ are reverse of each other, we denote $e_2$ by $\rev{e_1}$ and $e_1$ by $\rev{e_2}$.
	
A $g$-genus surface is a surface with $g$ many holes in it. Genus of a graph $G$ is the smallest $g$ such that $G$ can be embedded on a surface of genus $g$. For simplicity of discussion, we only consider graphs that can be embedded on an \emph{orientable} surface. Let $G$ be a graph embedded on a surface $S$ of genus $g$. The \emph{faces} of the embedding of $G$ are the connected components of $S\setminus G$. If a face is homeomorphic to an open disk, it is called a $2$-cell. If every face is homeomorphic to an open disk, the embedding is called a $2$-cell embedding. A combinatorial embedding of $G$ is defined as $\pi = \{ \pi_v \mid v \in V(G)\}$ where for each vertex $v$, $\pi_v$ is a cyclic permutation of darts whose tail is $v$. For a dart $e$, we use $\lface{e}$ to denote the face which is on the left of $e$ and $\rface{e}$ to denote the face on the right of $e$. A \emph{triangulated graph} is a graph that is embedded on a surface such that every face is a $2$-cell and has three boundary edges.
	
	We define a dual graph $\tilde{G}$ of $G$ for an embedding in the following way: $\tilde{G}$ contains a vertex $\tilde{v}$ corresponding to every face of $G$ and two vertices of $\tilde{G}$ have an edge between them if their corresponding faces share an edge in $G$. We say that an edge $\tilde e$ of $\tilde G$ crosses an edge $e$ of $G$ if the faces at the endpoints of $\tilde{e}$ shares the edge $e$ of $G$.
	
	Let $G$ be a graph embedded on a surface of genus $g$. Let $U$ be a subset of vertices of $G$, $F$ be a subset of edges of $G$ and $R$ be a subset of faces of $G$. Then $G[U]$ denotes the subgraph of $G$, induced by the vertices in the set $U$. Similarly, $G[F]$ denotes the subgraph of $G$, containing all the edges of $F$ together with their endpoints. By $G[R]$ we denote the graph containing all the vertices and edges that are in the boundary of a face in $R$.
	
	Let $G$ be a graph of genus $g$. A set $R$ of its faces is a \emph{region} if $\tilde{G}[R]$ is connected. The set of edges of $G$ whose only one side has a face in $R$ is called the \emph{boundary} of $R$.
	
	If $G$ is a graph embedded on a surface and $c$ is a cycle in $G$, then we define the \emph{left graph} and \emph{right graph} of $c$ as follows: If $e$ is a dart of $c$ followed by $e' = \pi^k_{\head{e}}(\rev{e})$, then all edges $\pi_{\head{e}}(\rev{e}), \pi^2_{\head{e}}(\rev{e}), \ldots, \pi^{k-1}_{\head{e}}(\rev{e})$ are said to be on the \emph{left} side of $c$. An edge $e''$ which is not incident with $c$ and which is connected by a path in $G\setminus c$ to an end of an edge of the left side of $c$ is also said to be on the left side. Now the left graph of $c$ is defined as the edges on the left side of $c$ together with all their ends. The right graph $G$ is defined analogously. We will often use the term \emph{inside of $c$} to denote the left graph of $c$ and \emph{outside of $c$} to denote the right graph of $c$.
	
	\begin{definition}
		Let $c$ be a cycle of an embedded graph $G$ such that the one of the sides of $c$ is planar. We call $c$ a \emph{contractible} cycle of $G$.
	\end{definition}
	
	A cycle which is not a contractible cycle is called a non-contractible cycle.
	
	We say that a set $C$ of cycles satisfies the $3$-path-condition if the following property holds: If $u$ and $v$ are vertices of $G$ and $P_1$, $P_2$ and $P_3$ are internally vertex disjoint paths from $u$ to $v$. If two of the cycles $C_{i,j} = P_i \cup P_j (1 \leq i < j \leq 3 )$ are not in $C$ then the third one is also not in $C$. It is a well known fact that the set of non-contractible cycles satisfies the $3$-path-condition \cite{Thomassen1990}.
	
	We define $\dist{u}{v}$ to be the length of the shortest path between two vertices $u$ and $v$. We introduce a total order (denoted by $<_v$) in the vertex set $V$ of the graph based on the distance from $v$. For any vertices $u$ and $w$, we say that $u$ is nearer to $v$ than $w$ (written as $u <_v w$) if we have either 
	\begin{itemize}
		\item $\dist{u}{v} < \dist{w}{v}$ or
		\item $\dist{u}{v} = \dist{w}{v}$ and $u$ has a smaller index than $w$
	\end{itemize}
	
	For any set $W$ of vertices of $G$, $\nrst_{v}(W)$ denotes a vertex $u$ in $W$ such that $u <_{v} w$ of all $w \in W \setminus \{u\}$. For any sets $W$ and $W'$ of vertices, we write $W <_{v} W'$ if $\nrst_{v}(W) <_{v} \nrst_{v}(W')$.
	
	Let $G$ be a graph of genus $g$. A closed loop $c$ is a sequence of \emph{distinct} darts $e_1, e_2, \ldots, e_m$ of $G$ such that $\head{e_{i}} = \tail{e_{(i+1)\bmod m}}$.
	
	\begin{definition}
		Let $G$ be a weighted-graph with positive integral weights on each of its vertices that sums to $n$ and $\alpha \in (0, 1)$. An $\alpha$-separator of $G$ is a set $S$ of vertices of $G$ such that the removal of $S$ creates disconnected subgraphs each of which has at most $\alpha n$ weight, where the weight of a subgraph is the sum of the weights of the vertices in it.
	\end{definition}
	
	We note that Allender and Mahajan \cite{AllenderMahajan} showed that the problem of testing whether a graph is planar or not is in $\SL$. They also gave the $\SL$ algorithm to construct the planar embedding. Subsequently, Reingold \cite{Reingold} showed that $\SL = \L$ hence there exists a logspace algorithm to test if a graph is planar and also produce its embedding. We summarize this fact in the following lemma.
	
	\begin{lemma}
		\label{lem:planaritytesting}
		There exists a logspace algorithm which tests whether the input graph is planar and if so, it outputs an embedding of the input graph.
	\end{lemma}
	
	Gilbert, Hutchinson and Tarjan proved the existence of an $O(n^{1/2}g^{1/2})$ size separator for the graphs of genus $g$ \cite{Gilbert}. They also presented an $O(n + g)$ time algorithm to find the separator. Therefore we can conclude that their algorithm runs in $O((n + g)\log n)$ space as well. We can thus use the following lemma for our result.
	
	\begin{lemma}
		\label{lem:genussep}
		There exists a polynomial time algorithm that takes, as an input, an $n$-vertex graph of genus $g$ along with its combinatorial embedding and finds its separator of size $O(n^{1/2}g^{1/2})$ using $O((n + g)\log n)$ space.
	\end{lemma}
	
We will need the notion of fundamental cycles in our construction of separator therefore we define it formally. 
	
	\begin{definition}
		Let $G$ be a graph and $T$ be a spanning tree of $G$. Let $e$ be an edge that does not belong to $T$. A simple cycle $c$, which consists of $e$ and the path in $T$ joining the endpoints of $e$ is called a fundamental cycle.
	\end{definition}
	
	\section{Voronoi Region}
	$\label{sec:voronoi}$
	
	As we discussed in section \ref{sec:intro}, we start by dividing the input graph into something that we call \emph{Voronoi regions}. In this section, we define the notion of Voronoi Regions and explain how they could be constructed in space-efficient manner. This notion has been previously used in designing a separator for planar graphs by Imai et al., Ashida et al., Gazit and Miller, and Koutis and Miller \cite{Imai, Imaiplanarsep, garrymiller, koutismiller}.
	
	We first define $k$-neighbourhood of a vertex. This is a key tool that will help us define and construct a Voronoi region.
	\begin{definition} Let $G$ be a graph and $v$ be a vertex of $G$. Let $L(v,i)$ be the set of vertices at distance $i$ from $v$. The $k$-neighbourhood $N_k(v)$ of a vertex $v$ is defined as:
		\[N_k(v) = \bigcup_{1 \leq i \leq d} L(v,i)\]    
		
		where $d$ is the smallest integer such that $\lvert\bigcup_{1 \leq i \leq d}L(v,i)\rvert \geq k$.
		
	\end{definition}
	
	Note that we have defined $k$-neighbourhood in a slightly different way when compared to the definition of Imai et al. \cite{Imai} and Chakraborty et al. \cite{Chakraborty}. In their work, $N_k(v)$ is chosen so that it contains at most $k$ vertices, while here it contains at least $k$ vertices. We believe this definition makes our proof simpler to follow.
	
	\begin{definition}Let $G$ be a graph. A set $I$ of vertices of $G$ is called a $k$-maximal independent set if the following holds:
		\begin{itemize}
			\item For every $b_1,b_2 \in I$, $N_k(b_1) \cap N_k(b_2) = \emptyset$.
			\item For every $v$ that is not in $I$, we have a vertex $b \in I$ such that $N_k(v) \cap N_k(b) \neq \emptyset$.
		\end{itemize}
	\end{definition}
	
	\begin{lemma}
		\label{lem:maxind}
		There exists an $O((k + n/k)\log n)$-space and polynomial time algorithm that takes a graph $G$ as input and outputs a $k$-maximal independent set $I$.
	\end{lemma}
	The proof of the above lemma is quite straight forward. We refer readers to \cite{Imai,fsttcstewari}.\\
	For a graph $G$, we will use the notation $\ind{G}$ to denote the set returned by the algorithm of lemma $\ref{lem:maxind}$.
	
	\begin{definition} Let $G$ be a graph. For any vertex $v$, the \emph{boss-vertex} of $v$ is a vertex $b$ of $\ind{G}$ such that $N_k(b) <_{v} N_k(b')$, for all $b' \in \ind{G}\setminus\{b\}$. We define $\vor{b}$ to be the set of all vertices whose boss-vertex is $b$. We use $\boss{v}$ to denote the boss-vertex of $v$.
	\end{definition}
	
	Note that the graph induced by the vertices in the set $\vor{b}$ form a connected component in $G$. Therefore, the faces corresponding to these vertices form a region in $\tilde{G}$. We will henceforth call $\vor{b}$ the \emph{Voronoi region} of $b$.
	
	The input graph might contain small non-contractible cycles. We require that the union of any two Voronoi regions do not have a non-contractible cycle, similarly as Chakraborty et al. \cite{fsttcstewari}. Thus, we remove such non-contractible cycles from the graph using the following lemma in our main algorithm.
	
	\begin{lemma}[\cite{fsttcstewari}]
		\label{lem:noncontr}
		There is an $O((k + n/k)\log n)$-space and polynomial time algorithm that takes a graph $G$, and two boss-vertices ${b}_1$ and ${b}_2$ as input and checks for a non-contractible cycle of size $O(k)$ in $\vor{{b}_1} \cup \vor{{b}_2}$. The algorithm outputs one such cycle if it exists.
	\end{lemma}
	
	\begin{proof}    
		First, consider the case when $\vor{{b}_1} \cup \vor{{b}_2}$ forms a connected subgraph of ${G}$. We combine the BFS-trees of $\vor{{b}_1}$ and
		$\vor{{b}_2}$ using an arbitrary edge to get a spanning tree of $\vor{{b}_1} \cup \vor{{b}_2}$ with diameter $O(k)$. We denote this spanning tree as $T$. Note that $T$ can be computed in polynomial-time and $O((n/k + k)\log n)$ space. We know that the set of all non-contractible cycles of any graph $G$ satisfy $3$-path condition \cite{Thomassen1990}. Since the diameter of $T$ is $O(k)$, any fundamental cycle of this tree of size $O(k)$. The 3-path condition implies that if a non-contractible cycle exists, then one of the fundamental cycles is non-contractible. We can check whether a cycle is contractible by checking planarity. Thus, we can do it in $O(\log n)$ space by Lemma \ref{lem:planaritytesting}. Thus, the lemma follows. The other case where $\vor{{b_1}}$ and $\vor{{b_1}}$ are not connected, we can apply the same procedure on spanning trees of $\vor{{b_1}}$ and $\vor{{b_2}}$ separately.\end{proof}
	
	As mentioned in the introduction, we will use the Voronoi regions to construct our Frame graph. For this construction, we first divide the Voronoi Regions.
	
	\subsection{Dividing Voronoi Regions using Pre-Frame-Loops}
	$\label{sec:dividingVornoi}$ 
	In this subsection, we find a set of loops in the input graph $G$. Each of these loops contains vertices of at most two Voronoi regions \emph{inside} them. We then further process these loops so that the number of vertices inside them is small.
	
	Let $G$ be a triangulated graph of genus $g$. Note that, any connected component of $G$ forms a region in $\tilde{G}$. Also, note that since the size of each face of $G$ is three, all the vertices of the graph $\tilde{G}$ will have degree three. Thus, a region of faces in $\tilde{G}$ will have a boundary that is a set of vertex-disjoint simple cycles.

 We require two kinds of edges in the dual graph to construct the desired loops. One is the set of the boundary edges of all the Voronoi regions, and the other is the set of \emph{Ridge edges}. Ridge edges have been used previously by Gazit and Miller \cite{garrymiller} and Ashida et al. \cite{Imaiplanarsep}. We define it as follows.
	
	\begin{definition}Let $G$ be a graph and ${b}$ be a boss-vertex. Let ${T}$ be the BFS tree of $\vor{{b}}$. For an edge ${e}$ of $G[\vor{{b}}]$ that do not belong to ${T}$, let ${c}_{e}$ be the fundamental cycle induced by ${e}$ on ${T}$. If each of the two sides of the cycle ${c}_{e}$ contains at least one boundary cycle of $\vor{{b}}$, then the edge $\tilde e$ of $\tilde G$ crossing ${e}$ is called a ridge edge.
	\end{definition}
\begin{figure}
	
	\centering
	\includegraphics[width = 0.3\textwidth]{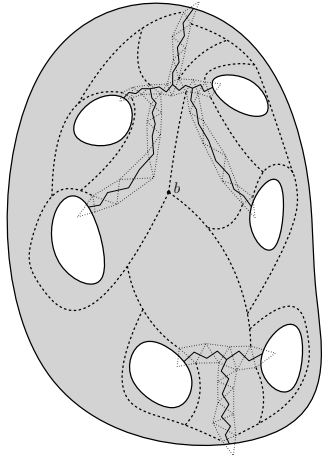}
	\caption{A diagram showing $\vor{b$} for a boss vertex $b$. The part of the surface where the vertices of $\vor{b}$ are present is shown in grey colour. The boundary of the Voronoi region is shown using thick solid lines. The ridge edges are shown using normal solid lines. Dashed lines show some of the edges of the spanning tree of $\vor{b}$. Dotted lines show faces in $G$ that corresponds to a vertex $v$ which is an endpoint of a ridge edge in $\tilde{G}$.} 
	\label{fig:vorb}
\end{figure}

Figure \ref{fig:vorb} shows ridge edges in the Voronoi region of a boss vertex $b$.
	
	\begin{definition}Let $G$ be a graph of genus $g$. Let $\tilde B$ be the set of boundary edges of $\vor{{b}}$ for all boss vertices $b$. Similarly, let $\tilde{R}$ be the set of ridge-edges. A \emph{branch vertex} is a degree three vertex in the graph $\tilde{G}[\tilde{B} \cup \tilde{R}]$. For each branch vertex $\tilde{v}$, the boundary of face consisting of three vertices incident to $\tilde{v}$ is a \emph{branch-triangle}. Two branch vertices are called \emph{adjacent} to each other if they are connected by a path consists of darts corresponding to the edges in the set $\tilde B\cup \tilde R$ such that no other branch vertex exists on this path. The path connecting adjacent branch vertices is called a \emph{connector}. We denote the set of connectors in $\tilde G$ by $\con{G}$.
	\end{definition}
	
	Let $\tilde p$ be a connector. Note that the end points of $\tilde{p}$ are an adjacent pair of branch vertices. Also note that, $\boss{\lface{\first{p}}}$ is same as $\boss{\lface{\last{p}}}$ and $\boss{\rface{\first{p}}}$ is same as $\boss{\rface{\last{p}}}$. We define a pre-frame-loop with respect to $\tilde p$ as follow.
	
	\begin{definition}Let $G$ be a graph embedded on a surface of genus $g$. For any connector $\tilde p$, a pre-frame-loop (denoted by $\pfloop{\tilde p}$) is a closed loop that consists of 
		\begin{enumerate}
			\item A path from $\rface{\first{p}}$ to $\boss{\rface{\first{p}}}$ in the BFS-tree of $\vor{\boss{\rface{\first{p}}}}$
			\item A path from $\boss{\rface{\first{p}}}$ to $\rface{\last{p}}$ in the BFS-tree of $\vor{\boss{\rface{\first{p}}}}$ 
			\item A branch-triangle dart ${e}_{last}$ from $\rface{\last{p}}$ to $\lface{\last{p}}$
			\item A path from $\lface{\last{p}}$ to $\boss{\lface{\last{p}}}$ in the BFS-tree of $\vor{\boss{\lface{\last{p}}}}$.
			\item A path from $\boss{\lface{\last{p}}}$ to $\lface{\first{p}}$ in the BFS-tree of $\vor{\boss{\lface{\last{p}}}}$. 
			\item A branch-triangle dart ${e}_{fst}$ from $\lface{\first{p}}$
			to $\rface{\first{p}}$.
		\end{enumerate}
		We denote the set of all the pre-frame-loop in $G$ as $\pfloop{G}$    
	\end{definition}
	
	\begin{lemma}
		\label{lem:pfcycle}
		There exists an $O((k + n/k)\log n)$-space and polynomial time algorithm that takes $G$ as an input and outputs the list $\pfcycle{ G}$ of all pre-frame-loops in ${G}$.
	\end{lemma}
	
	In the next section, we will use these pre-frame-loops to create faces of our subgraph. Following Ashida et al. \cite{Imaiplanarsep}, we call this new graph Frame Graph.
	
	\section{Frame Graph}
	\label{sec:framegraph}
	We wish to use pre-frame-loops to create faces of the frame graph. In order to do this, we first preprocess these loops so that the \emph{inside} of each loop is small, i.e., has at most $n/3$ vertices in it. This would ensure that the weight on any face of the frame graph is bounded. In the second step, we remove those edges of the loop for which both of its darts are traversed and thus break the loop into simple cycles. These cycles will act as boundaries of the faces in the frame graph.
		
	Consider a connector $\tilde {p}$ of the input graph $G$ and the pre-frame-loop $c$ induced by $\tilde{p}$. Note that $c$ is in the union of two Voronoi regions. Since we have eliminated all non-contractible cycle from the union of any two Voronoi regions, $c$ cannot contain a non-contractible cycle. Thus, $c$ divides the surface. A pre-frame loop is of type $A$ if it consists of two boss vertices and it is of type $B$ if it consists of only one boss vertex (see Figure \ref{fig:type}).    
	
	\begin{figure}
		\begin{center}
			\begin{tikzpicture}[]
			\usetikzlibrary{decorations.markings}
			\begin{scope}[very thick,decoration={
				markings,
				mark=at position 0.5 with {\arrow{>}}}
			]
			
			\draw[postaction={decorate}][draw, thick] (0,0) -- (4,0);
			
			\draw[postaction={decorate}][draw, thick] (7,0) -- (11,0);

			\end{scope}

			\draw[dashed,thick] (.5,-.5) to [] (.5,.5);
			\draw[dashed,thick] (3.5,-.5) to [] (3.5,.5);
			
			\draw[dashed,thick] (.5,.5) to [] (2,2);
			\draw[dashed,thick] (3.5,.5) to [] (2,2);
			
			\draw[dashed,thick] (.5,-.5) to [] (2,-2);
			\draw[dashed,thick] (3.5,-.5) to [] (2,-2);

			\draw[fill] (2,2) circle [radius=0.1];
			\draw[fill] (2,-2) circle [radius=0.1];
			
			\draw[] (.5,-.5) circle [radius=0.05];
			\draw[] (.5,.5) circle [radius=0.05];
			\draw[] (3.5,-.5) circle [radius=0.05];
			\draw[] (3.5,.5) circle [radius=0.05];
			\draw[] (0,0) circle [radius=0.05];
			\draw[] (4,0) circle [radius=0.05];
			\node at (2,-.4) {$\tilde{p}$};

			\draw[dashed,thick] (7.5,-.5) to [] (7.5,.5);
			\draw[dashed,thick] (10.5,-.5) to [] (10.5,.5);
			
			\draw[dashed,thick] (7.5,.5) to [] (9,2);
			\draw[dashed,thick] (10.5,.5) to [] (9,2);

			\draw[dashed,thick] (7.5,-.5) .. controls (6,0) .. (9,2);
			\draw[dashed,thick] (10.5,-.5) .. controls (12,0) .. (9,2);

			\draw[fill] (9,2) circle [radius=0.1];

			\draw[] (7.5,-.5) circle [radius=0.05];
			\draw[] (7.5,.5) circle [radius=0.05];
			\draw[] (10.5,-.5) circle [radius=0.05];
			\draw[] (10.5,.5) circle [radius=0.05];
			\draw[] (7,0) circle [radius=0.05];
			\draw[] (11,0) circle [radius=0.05];
			\node at (9,-.4) {$\tilde{p}$};
			\node at (2,2.3) {$b_1$};
			\node at (2,-2.3) {$b_2$};
			\node at (9,2.3) {$b_1$};
			\node at (2,1) {$P_0$};
			\node at (1,2) {$P_1$};
			\node at (7.15,.25) {$P_{1,1}$};
			\node at (11,.25) {$P_{1,2}$};
			\node at (9,1) {$P_{0}$};
			\node at (9,3) {$\textrm{Type B}$};
			\node at (2,3) {$\textrm{Type A}$};
			\end{tikzpicture} 
			\caption{On the left, a pre-frame-loop of type $A$. The two boss vertices corresponding to the loop are $b_1$ and $b_2$. On the right, a pre-frame-loop of type $B$. The only boss vertex corresponding to this loop is $b_1$}
			\label{fig:type}
		\end{center}
	\end{figure}
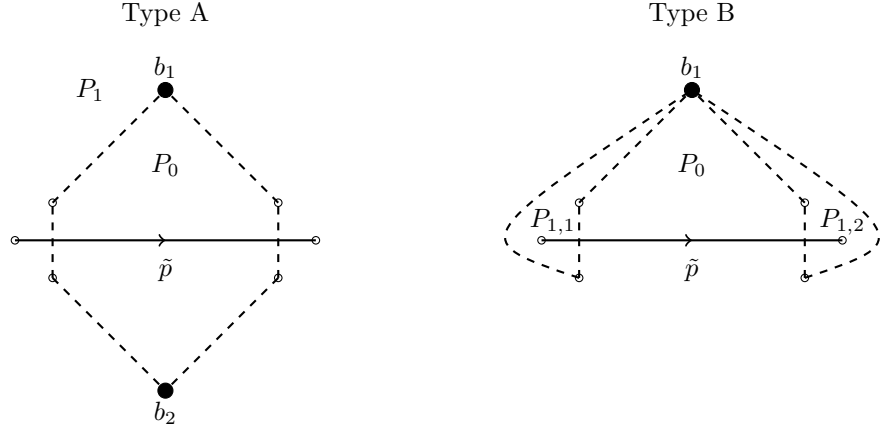
	
	Let the part of a connector $\tilde{p}$ excluding its first and last vertex be called the \emph{body} of $\tilde{p}$. Let $P_0$ denote the surface of $G \setminus c$ that has the body of the connector $\tilde{p}$. Call this \emph{inside} of ${c}$. Let ${n}_0$ be the number of vertices in $P_0$ not including the vertices of ${c}$. We say that the inside of $c$ is \emph{large} if $n_0$ is greater than $2n/3$. Note that the inside of ${c}$ is included in the union of atmost two Voronoi regions. Let ${b}_1$ and ${b}_2$ be the boss vertices of these two regions. We use the BFS-Trees of $\vor{{b}_1}$ and $\vor{{b}_2}$ to find a spanning tree of $\vor{{b}_1} \cup \vor{{b}_2}$. Since $\vor{{b}_1} \cup \vor{{b}_2}$ does not have a non-contractible cycle, it has a planar embedding. Hence we can use the spanning tree of $\vor{{b}_1} \cup \vor{{b}_2}$, along with the result algorithm of Lipton and Tarjan to find the separator of the graph $G[\vor{{b}_1} \cup \vor{{b}_2}]$. We can combine this separator with the boundary of the pre-frame-loop to get a separator of the whole graph. We summarize this fact in the following lemma.
	
	\begin{lemma}
		\label{lem:largein}
		Let $G$ be a graph of genus $g$ and $c$ be a pre-frame loop of $G$ whose inside is large. There exists a polynomial time algorithm that takes as an input $G$ and $c$ and outputs a separator of $G$ of size $O(gn^{1/2})$ in $O(gn)^{1/2}\log n)$ space.
	\end{lemma}
	
	Thus, if any of the pre-frame-loop acts as a separator or has a large inside, we can get a separator of the graph $G$. Otherwise, we construct a set $C$ in the following way: We first add all the pre-frame-loop of type $A$ into $C$.
	Note that if a pre-frame-loop $c$ is of type $B$ pre-frame-loop, it divides the surface into three parts. Call the two parts of the surface which does not contain the body of the connector as $P_{1, 1}$ and $P_{1, 2}$ respectively. Let the number of vertices in $P_0$, $P_{1, 1}$ and $P_{1, 2}$ be ${n}_0$, ${n}_{1, 1}$ and ${n}_{1, 2}$ respectively.
	We see that either ${n}_{1, 1} > 2{n}/3$ or ${n}_{1,2} > 2{n}/3$, for otherwise, our pre-frame-loop acts as a separator. Let us assume, without loss of generality, that ${n}_{1, 2} > 2{n}/3$. We merge  $P_0$ and $P_{1, 1}$ into a single surface, and add the loop ${c}_0$ bounding this surface to the set $C$. The $\emph{inside}$ of $c_0$ is the side containing the surfaces $P_0$ and $P_{1,1}$.
	
	Now, consider a loop ${c}$ of $C$, that is not contained in the inside of any other loop $c$ of $C$. Let ${E}_{{c}}$ be the set of darts whose reverse does not appear in ${c}$. Let $E$ be the union of $E_{c}$ over all such $c$. We observe that $E$ is a set of simple cycles which we call frame-cycles and denote by $\fcycle{G}$.
	
	\subsection{Definition and construction of Frame Graph}
	
	\begin{definition}Let $G$ be a graph of genus $g$. Let ${E}_1$ be the set of all frame-cycles edges, and let ${E}_2$ be the set of all branch-triangle edges. A \emph{frame-graph} of $G$ is a subgraph ${H} = G[E_1 \cup E_2]$. For each face of a frame-graph ${H}$, its weight is the number of vertices of ${G}$ located inside that face. We denote the frame graph of $G$ by $\fgraph{G}$.
	\end{definition}
	
	\begin{definition}
		Let $G$ be a triangulated graph. Let $L(v,i)$ be the set of vertices at distance $i$ from $v$. Let $d_{nb}(v)$ be the largest $d$ such that $\lvert \cup_{0\leq i\leq d} L(v,i)\rvert < k$. For any boss-vertex ${b} \in \ind{G}$, let $d_{core}({b})$ denote the largest $d \leq d_{nb}({b})$ such that $\lvert L({b}, d)\rvert \leq  k^{1/2}$. The \emph{core} of ${b}$ (denoted by $\core{{b}}$) is defined by
		\[\core{ b} = \bigcup_{0\leq i\leq d_{core}( b)}L({b}, i)\]    
	\end{definition}

Note that core(b) forms a region in $\tilde{G}$. The boundary of this region might not be a single cycle. In the next definition, we pick one of these cycles to be core-boundary-cycle and use it to construct the core cycle in the graph $G$.
	
\begin{definition}
		Let $G$ be a triangulated graph of genus $g$ and $b$ be a boss-vertex. The \emph{core-boundary-cycle} of $b$ is the boundary cycle of the region $\core{{b}}$ in $\tilde{G}$ that has the largest number of dual-vertices on its outside.
		
		The \emph{core-cycle} of $\core{{b}}$ is a directed cycle induced by the set of vertices in $\core{{b}}$ sharing an edge with the core boundary cycle. The inside of the core-cycle is the side with the boss-vertex ${b}$.
	\end{definition}
	
	For any $l \geq 1$, let ${L}_{nb}(l)$ denote a
	set of vertices $v$ of $G$ whose distance from its nearest $k$-neighborhood in $\{N_k( b)\}_{ b \in \ind{G}}$ is $l$. More formally,	
	\[{L}_{nb}(l) = \{{v} \mid \dist{{v}}{{v}_{\nrst}} = l, \text{where     } {v}_{\nrst} = \nrst_{ v}(N_k(\boss{v}))\}.
	\]
	
Let $\tilde{{L}}_{nb}(l)$ denotes the set of faces in $\tilde{G}$ corresponding to the vertices in the set ${L}_{nb}(l)$. Let $C$ be a region of $\tilde{{L}}_{nb}(l)$. Each boundary edge of ${C}$ is an edge between a pair of vertices of level either $l-1$ and $l$ or $l$ and $l + 1$. Let us call the former one an \emph{interior} edge and the latter one an \emph{exterior} edge. We call a boundary cycle an \emph{interior boundary cycle} if it consists of interior edges. Similarly, we call a boundary cycle an \emph{exterior boundary cycle} if it consists of exterior edges.
	
	\begin{definition}
		Let $\tilde{c}$ be any interior boundary cycle corresponding to ${L}_{nb}(l)$. Let $c$ be the loop in $C$ formed by the set of vertices sharing a boundary edge with $\tilde{c}$, and let $D_c$ be the set of cycles obtained from $c$ by removing all the darts in the loop whose reverse also appears in the loop. An {interior-cycle} is a cycle in $D_c$. 
	\end{definition}
	
	We define an $\emph{exterior-cycle}$ in a similar way. A cycle is said to be a \emph{small} cycle if it consists of at most $k^{1/2}$ vertices. We denote the set of small interior cycles by $\smint{G}$ and the set of small exterior cycles by $\smext{G}$. A contractible cycle is said to be \emph{light} if it has less than ${n}/3$ vertices in its inside.
	
	\begin{definition}
		A floor-cycle is a light and small interior-cycle if it is not in the inside of any other light and small interior-cycle. For any boss-vertex ${b}$ which is not contained in any floor-cycle, we regard the core-cycle of $\core{{b}}$ also as a floor-cycle. A ceiling-cycle is a light and small exterior-cycle which is not in the inside of any other light and small exterior-cycle, and that has at least one dual-vertex of some branch-triangle on its inside.
	\end{definition}
	
	\begin{definition}Let $G$ be a graph of genus $g$. Let $F$ and $C$ be respectively a set of floor-cycles and ceiling-cycles having at least one vertex of $\fgraph{G}$ in their insides. Let ${E}'_1$ be the set of edges of $G$ that appear in some cycle in ${F} \cup {C}$ and ${E}'_2$ be the set of edges of $\fgraph{G}$ that are not in the inside of any cycle of ${F} \cup {C}$. A graph with vertices $U'$ and edges $D'$ is a modified frame-graph denoted as $\mframe{G}$, where $D' = E'_1 \cup E'_2$ and $U'$ is the set of all vertices that are endpoints of edges of ${D}'$. For each face of a modified frame-graph $\mframe{G}$, its weight is the number of vertices of ${G}$ located in the face.
	\end{definition}
	
	The following lemma is a generalization of a result that was presented by Ashida et al. \cite{Imaiplanarsep}. They presented a similar lemma for planar graphs. 
	
	\begin{lemma}\label{lem:imp}Let $G$ be a graph of genus $g$ such that voronoi region $\vor{b_1} \cup \vor{b_2}$ does not contain a non-contractible cycle for any two vertices $b_1, b_2 \in \ind{G}$, $\pfloop{\tilde p}$ is not a separator of $G$ for any connector $\tilde p$, the inside of any loop in $\pfloop{\tilde{p}}$ is not large, $\core{b}$ is not a separator of $G$ for any boss-vertex $b$, and no cycle in $\smext{G}$ or $\smint{G}$ is a non-contractible cycle. Following statements hold:
		\begin{enumerate}
			\item The weight of each face of $\mframe{G}$ is less than $n/3$.
			\item $\mframe{G}$ is $2$-connected.
			\item Size of each face of $\mframe{G}$ is $O(k^{1/2})$.
			\item The number of faces in $\mframe{G}$ is $O(n/k+g)$
		\end{enumerate} 
	\end{lemma}
\begin{proof}[Proof of Lemma \ref{lem:imp}] We will prove each of the four statements of the proof in order.
	\begin{enumerate}
		\item Consider a face of the frame-graph $\fgraph{G}$. The boundary of this face is either a frame-cycle or a branch-triangle. The weight of a branch-triangle is zero, while the number of vertices inside a frame-cycle is less than $n/3$ by construction. Hence the weight of any face of $\fgraph{G}$ is less than $n/3$. The number of vertices inside a floor or a ceiling cycle is less than $n/3$ by definition. Hence the weight of any face of $\mframe{G}$ is also less than $n/3$.
		\item We first prove that $\fgraph{G}$ is $2$-connected. Let $u$ and $v$ be two distinct vertices of $\fgraph{G}$. We have the following cases:
		\begin{description}
			\item[Case 1 (Both $u$ and $v$ are on same frame cycle in $\pfcycle{G}$) :] Since $u$ and $v$ are on a cycle, there exist two vertex-disjoint paths from $u$ to $v$.
			\item[Case 2 ($u$ and $v$ are on two different cycles in $\pfcycle{G}$) :] Let $c_u$ and $c_v$ be the cycles of $\pfcycle{G}$ which contain vertices $u$ and $v$ respectively. Let $\tilde{p}_u$ and $\tilde{p}_v$ be the connectors whose bodies are contained in $c_u$ and $c_v$ respectively. We first note that there is a sequence of connecters $\tilde{p}_u = \tilde{p}_1, \tilde{p}_2, \ldots , \tilde{p}_k = \tilde{p}_v$ such that $\tilde{p}_i$ and $\tilde{p}_{i+1}$ has a common end point for each $i \in [1, k-1]$. Also note that the body of $\tilde{p}_i$ is contained in a cycle $c_i$ of $\pfcycle{G}$. Orient the darts of $\tilde{p}_i$ to form a path from $\first{p_1}$ to $\last{p}_k$. Let $c_i^{left}$ be the path from $\lface{\first{\tilde{p}_i}}$ to $\lface{\last{\tilde{p}_i}}$. Similarly, let $c_i^{right}$ be the path from $\rface{\first{\tilde{p}_i}}$ to $\rface{\last{\tilde{p}_i}}$. We see that $c_i^{left}$ and $c_i^{right}$ do not share any vertex. Thus, we can see that there exist two vertex-disjoint paths $q_{left}$ and $q_{right}$ from $u$ to $v$ such that $q_{left}$ contains vertices of $c_i^{left}$ and $q_{right}$ contains vertices of $c_i^{right}$ for all $i \in [2, k-1]$.\end{description}
		The analysis of other cases is similar.\\

		Now we will show that there exist two vertex disjoint paths between any two vertices $u$ and $v$ of the graph $\mframe G$.        
		\begin{description}
			\item[Case 1 ($u$ and $v$ are both in $\fgraph{G}$):] Since we have proved that $\fgraph{G}$ is two connected, we know that there exist two vertex disjoint paths $q_{left}$ and $q_{right}$ between $u$ and $v$ in $\fgraph{G}$. Note that several floor-cycles and ceiling-cycles were added to $\fgraph{G}$ and the vertices inside them were removed, in order to construct $\mframe G$. Let $c$ be a one such cycle.
			\begin{itemize}
				\item If $c$ intersects both $q_{left}$ and $q_{right}$. Let $u_{left}$ and $v_{left}$ denote the first and the last vertices of $q_{left}$ which intersects $c$. Similarly, let $u_{right}$ and $v_{right}$ denote the first and the last vertices of $q_{right}$ which intersects $c$. These four vertices divide $c$ into four paths $c_1, c_2, c_3$ and $c_4$. Let the set of these four paths be $C$. Then one of the following statements is true:
				\begin{itemize}
					\item There exist paths from $u_{left}$ to $v_{left}$ and from $u_{right}$ to $v_{right}$ in $C$.
					\item There exist paths from $u_{left}$ to $v_{right}$ and from $u_{right}$ to $v_{left}$ in $C$.
				\end{itemize}
				For both the above cases, we see that there exist two disjoint paths from $u$ to $v$.
				
				\item If $c$ intersects only one of the path $q_{left}$ and $q_{right}$ then we can modify that path to contain part of the cycle.
			\end{itemize}			
			
			\item[Case 2 ($u$ and $v$ are on different floor-cycles or ceiling-cycles $c_u$ and $c_v$):] Let $w_u$ and $w_v$ be vertices of $\fgraph{G}$ inside $c_u$ and $c_v$ respectively. We know such vertices exits because of the way these cycles are defined. Since the graph $\fgraph{G}$ is $2$-connected, there exist two disjoint paths $q_{left}$ and $q_{right}$ between $w_u$ and $w_v$ in it.
			Let $u_{left}$ be the last intersection of $q_{left}$ and $c_u$. Similarly let $u_{right}$ be the last intersection of $q_{right}$ and $c_u$. Note that, since the paths $q_{left}$ and $q_{right}$ are disjoint, $u_{left} \neq u_{right}$. We similarly define $v_{left}$ and $v_{right}$.
			
			Since the three vertices $u$, $u_{left}$ and $u_{right}$ lie on the cycle $c_u$, there exists two disjoint paths: first from $u$ to $u_{left}$ and second from $u$ to $u_{right}$. Similarly, there exists two disjoint paths from $v_{right}$ to $v$ and from $v_{left}$ to $v$. We can thus get two vertex-disjoint paths from $u$ to $v$, using these. Note that there may be other floor and ceiling cycles intersecting these disjoint paths. In that case, we can use an argument similar as above to show the existence of two disjoint paths from $u$ to $v$.
			
		\end{description}
		\item We now prove that the size of each face of $\mframe{G}$ is $O(k^{1/2})$. Note that the boundary of a face of the graph $\mframe G$ is one of the following:
		\begin{enumerate}
			\item A floor-cycle of $G$.
			\item A ceiling-cycle of $G$.
			\item A branch-triangle of $G$.
			\item A frame-cycle of $\fgraph{G}$ modified by floor-cycles and ceiling-cycles.
		\end{enumerate}
		In the first three cases, the size bound of the face follows by definition. We thus consider the fourth case.
		
		Consider any face defined by a modified frame-cycle, and let $c$ denote the pre-frame-loop from which we have defined it. Consider any path $p$ of $c$ connecting a boss-vertex of $c$ and a vertex of a branch-triangle used in $c$ such that the path does not contain any other vertex of a branch-triangle. By our modification, we can use a part $p'$ of $p$ that is in the outside of the corresponding floor-cycle and ceiling-cycle (if it exists) as a component of the modified frame-cycle, and its length is bounded by $4k^{1/2}$. Note that the floor-cycle may not be used in $\mframe{G}$ if it only intersects with the darts that we have removed for defining the face. In this case, however, only a part of $p'$ is used for the modified frame-cycle, which is even shorter. Thus, the modified frame-cycle consists of at most four such reduced paths, a part of two floor-cycles, a part of four ceiling-cycles, and two edges from two branch-triangles, and their total length is $O(k^{1/2})$.
		
		\item We first prove that the number of connectors is $O(n/k +g)$, and the number of branch vertices is $O(n/k+g)$. Since there is a one-to-one correspondence respectively between
		branch-triangles and branch vertices, and between frame-cycles and connectors, it will follow that the number of faces in $\fgraph{G}$ is $O(n/k + g)$.
		
		We first define a new graph $G'$. The vertex set of $G'$ is the set of branch vertices in $G$. We add an edge between two vertices of $G'$ if they are {adjacent} pair of branch vertices. Since every edge of $G'$ corresponds to a connector of $G$, the graph $G'$ can also be embedded on the surface of genus $g$ where the embedding corresponds to the embedding of $G$. Let $n'$, $e'$, $f'$ be the number of vertices, edges and faces in $G'$ respectively. Thus, we have $n' - e' + f' = 2-2g$ by Euler's formula. Note that every branch vertex have a degree $3$, therefore we have $2e' = 3n'$. This implies $e' = 6g + 3f' - 6 = O(f' + g)$. Since, there is one-to-one correspondence between voronoi regions and the faces of $G'$, we have $f' = O(n/k)$. Hence, we can conclude that $e' = O(n/k + g)$ and $n' = O(n/k + g)$. 
		
		Now, to prove that the number of faces in $\mframe{G}$ is $O(n/k+g)$ we see that the new faces introduced by our modification are those defined by floor-cycles or ceiling-cycles. By definition, the number of these cycles is at most the number of boss-vertices or that of branch-triangles, which is bounded by $O(n/k + g)$. Note that we can divide a face defined by a frame-cycle of $\fgraph{G}$ by ceiling-cycles, but it is easy to see that each face is divided into at most some constant number of faces because the number of floor-cycles and ceiling-cycles overlapping each frame-cycle is constant, say, at most six. From these observations, we can bound the number of faces of $\mframe{G}$ by $O(n/k + g)$.\qedhere\end{enumerate}
\end{proof}

	\section{Construction of separator}
	\label{sec:final}

Using the tools developed so far, we can obtain a space-efficient algorithm which given a graph $G$ as input, outputs either a separator, a non-contractible cycle or the modified frame graph $\mframe G$. We summarize this in the following lemma. 
	
	\begin{lemma}\label{lem:main} Let $G$ be a $g$-genus triangulated graph of $n$ vertices. For any positive integer $k$, there is a polynomial time, $O((n/k + k)\log n)$-space algorithm that takes $G$ along with its combinatorial embedding as input and outputs one of the following:
		\begin{enumerate}
			\item A non-contractible cycle of size $O(k)$ of ${G}$.
			\item A separator of size $O(k)$ of ${G}$.
			\item A a weighted subgraph ${H}'$ of ${G}$ that satisfies the following conditions:
			\begin{enumerate}
				\item The weight of each face $f$ of ${H'}$ is proportional to the number ${n_f}$ of vertices of ${G}$ located inside the face, and is less than $n/3$
				\item ${H}'$ is 2-connected.
				\item $H'$ contains $O({n}/k + g)$ faces.
				\item The size of each face of ${H}'$ is $O(k^{1/2})$.
			\end{enumerate}
		\end{enumerate}
	\end{lemma}
	
	\begin{proof}
		We first find a $k$-maximal independent set $\ind{G}$ of ${G}$ in $O((n/k+k)\log n)$-space and polynomial time using lemma $\ref{lem:maxind}$. We then check for a non-contractible cycle in $\vor{{b}_i} \cup \vor{{b}_j}$ for all pairs of boss-vertices ${b}_i$ and ${b}_j$ using lemma $\ref{lem:noncontr}$. If we manage to find such a cycle, we output it. Otherwise, we pick each pre-frame loops using lemma $\ref{lem:pfcycle}$ see if it acts as a separator of the graph. If so, we output it.
		 
If the algorithm has not produced an output so far, we see if the inside of any pre-frame-loop is large. If so, we use lemma $\ref{lem:largein}$ to find a separator of the graph. For every boss vertex ${b}$, we check if $\core{{b}}$ is a separator. If so, we output it. Next, we check if any cycle in $\smext{G}$ or $\smint{G}$ is a non-contractible cycle. If so, we output it. Otherwise, we output the modified frame graph $\mframe{G}$.
	\end{proof}
	With these ingredients, we are now ready to prove our main theorem.
	\begin{proof}[Proof of Theorem \ref{thm:main}]
		Elberfeld and Kawarabayashi presented an algorithm to construct a combinatorial embedding of a graph of constant genus in logspace \cite{boundedgenusembedd}. Hence, when dealing with constant-genus graph, we do not require a combinatorial embedding as part of the input. Otherwise, we require the combinatorial embedding of the graph as an input. We assume that the genus of the input graph $g$ is at most $O(n)$. Let $\pi$ be the combinatorial embedding of $G$. We first triangulate the input graph in logspace. To do this, for each face $f$ of the input graph, we connect each vertex of $f$ with the lowest index vertex in it. This triangulation is done implicitly, whenever required, as storing the triangulated graph will require a large amount of space. We call the resultant triangulated graph $G$. Note that triangulating the graph only introduces more edges therefore, a separator for $G$ will also be a separator for the input graph. Our objective now is to construct a separator of $G$. We do this by iteratively applying lemma $\ref{lem:main}$. We will describe the algorithm by describing an iteration of it. Before the $i$th iteration, we will have a set $S$ of vertices which is empty before the first iteration. Let $G_1,G_2, \ldots , G_m$ be the set of connected components in $G \setminus S$. We will describe the $i$th iteration as follows. The algorithm takes the component $G_j$ whose size $n_j$ is greater than $2n/3$. If no such component exists, then the set $S$ would be a separator of $G$, and the algorithm outputs $S$ and halts. Otherwise, consider the embedding induced by $\pi$ on $G_j$ as its embedding. If the genus $g_j$ of this component is zero, the algorithm uses Imai et al. planar separator algorithm to get its separator $S_1$ and outputs $S \cup S_1$. If its genus is non-zero we apply the algorithm from Lemma $\ref{lem:main}$ on $G_j$ with $k$ set as $n_j^{1/2}/g_j^{1/2}$. If the result of the application of the algorithm from Lemma $\ref{lem:main}$ on $G_j$ is a non-contractible cycle say $S_2$, then we add the vertices of $S_2$ to the set $S$ and continue with the next iteration. If the result is a separator say $S_3$, we output the set $S \cup S_3$ as the separator for the entire graph. Otherwise, if the result is a subgraph ${H}'$ of $G_j$, we take its dual $\tilde{H}'$ and find its separator $\tilde S'$ using Lemma $\ref{lem:genussep}$. $\tilde{S}'$ is a set of faces of $H'$. Consider the set $S_4$ of vertices on the boundary of these faces. We return the set $S \cup S_4$. To see that the size of separator returned by the above algorithm is $O(g^{1/2}n^{1/2})$, note that Lemma $\ref{lem:main}$ returns a non-contractible cycle, the genus of the graph is reduced by at least one. Hence, it can return at most $gk$ such cycles. For our value of $k$ the total number of vertices in all such cycles can be at most $O(g^{1/2}n^{1/2})$. If it does not returns a non-contractible cycle, then it returns either a separator of size $O(k) \leq O(g^{1/2}n^{1/2})$ or it returns the subgraph $H'$. The number of faces in $H'$ is $O(n/k + g)$. Hence the size of the separator returned by using the algorithm of Gilbert et al. \cite{Gilbert} on the dual of $H'$ will be $O(g^{1/2}(n/k + g)^{1/2})$. Size of each face of $H'$ is at most $k^{1/2}$, hence, size of the set $S_4$ is $O(k^{1/2}g^{1/2}(n/k + g)^{1/2})$. For our value of $k$, this is at most $O(g^{1/2}n^{1/2})$.
	\end{proof}
	
	We can use the following lemma which was formalized by Jain and Tewari \cite{mytreewidth} to get a space-efficient polynomial time algorithm for reachability in constant-genus graphs. Reachability is the problem of determining if there is a directed path from one vertex to another in a directed graph.
	
	\begin{lemma}
		Let $\mathcal{G}$ be a class of graphs and $w: \mathcal{N} \mapsto \mathcal{N}$ be a function. If there exist a polynomial time algorithm that uses $O(w(n)\log n)$ space to find a separator of size $w(n)$ then there exists a polynomial time algorithm to decide reachability in $G$ that uses $O(w(n)\log n)$ space.
	\end{lemma}
	
	\begin{corollary}
		\label{cor:main}
		There exists a polynomial time algorithm that uses $O(n^{1/2}\log n)$ space to solve reachability in a constant-genus graph.
	\end{corollary}
	
	Previously, a polynomial time algorithm that uses $O(n^{1/2}\log n)$ space for reachability was known for planar graphs \cite{Imai}. While for constant-genus graphs, a polynomial-time algorithm that uses $O(n^{2/3}\log n)$ space was known \cite{fsttcstewari}. Corollary \ref{cor:main} improves the space-bound to $O(n^{1/2}\log n)$. Our result can thus be seen as both a generalization of Imai et al. $\cite{Imai}$ and as an improvement to a previous result by Chakraborty et al. $\cite{Chakraborty}$.
	
	\bibliography{biblio}
	
\end{document}